\documentclass[conference,a4paper]{IEEEtran}
\pdfoutput=1
\usepackage{pgfplots}
\pgfplotsset{compat=newest}
\usepackage{tikz}
\usetikzlibrary{arrows,matrix,positioning,patterns}
\usepackage[utf8]{inputenc}
\usepackage[innermargin=0.5in,outermargin=0.47in,top=0.51in,bottom=.51in]{geometry}
\usepackage[english]{babel}
\usepackage[T1]{fontenc}
\usepackage{epsfig}
\usepackage{amsmath, amssymb, amsbsy}
\usepackage{mathdots}
\usepackage[normalem]{ulem}
\usepackage{xspace}
\usepackage[noend]{algpseudocode}
\usepackage[linesnumbered,ruled,vlined,titlenumbered]{algorithm2e}
\usepackage{algorithmicx}
\usepackage{color}
\usepackage{cite}
\usepackage{booktabs}
\usepackage{verbatim}
\usepackage{url}
\usepackage{lipsum}
\usepackage{enumitem}
\usepackage{colortbl}
\usepackage{amsthm}
\usepackage{flushend}

\makeatletter
\newcommand\footnoteref[1]{\protected@xdef\@thefnmark{\ref{#1}}\@footnotemark}
\makeatother

\newtheorem{theorem}{Theorem}
\newtheorem{lemma}[theorem]{Lemma}

\newtheorem{remark}[theorem]{Remark}

\newtheorem{definition}{Definition}

\usepackage[final,tracking=true,kerning=true,spacing=true,factor=1100,stretch=10,shrink=20]{microtype}

\newenvironment{mymatrix}{\begin{bmatrix}} {\end{bmatrix} }

\def\ve#1{{\mathchoice{\mbox{\boldmath$\displaystyle #1$}}%
              {\mbox{\boldmath$\textstyle #1$}}%
              {\mbox{\boldmath$\scriptstyle #1$}}%
              {\mbox{\boldmath$\scriptscriptstyle #1$}}}}

\DeclareMathOperator{\extsmallfield}{ext}
\DeclareMathOperator{\rank}{rk}

\DeclareMathOperator{\ftemp}{f}

\renewcommand{\S}{\ve{S}}
\renewcommand{\H}{\ve{H}}

\renewcommand{\a}{\ve{a}}

\newcommand{\e}{\ve{e}}
\renewcommand{\v}{\ve{v}}

\renewcommand{\c}{\ve{c}}
\renewcommand{\e}{\ve{e}}

\newcommand{\A}{\ve{A}}
\newcommand{\B}{\ve{B}}

\newcommand{\D}{\ve{D}}

\newcommand{\NM}{\mathrm{NM}}
\newcommand{\ZZ}{\mathbb{Z}}
\newcommand{\frk}{\mathrm{frk}}
\newcommand{\rk}{\mathrm{rk}}
\newcommand{\T}{\ve{T}}

\newcommand{\Rq}{{R}_q}
\newcommand{\Rqm}{{R}_{q,m}}

\newcommand{\Fspace}{\mathcal{F}}
\renewcommand{\r}{\ve{r}}
\newcommand{\s}{\ve{s}}
\newcommand{\Sspace}{\mathcal{S}}
\newcommand{\Espace}{\mathcal{E}}
\newcommand{\supp}{\mathrm{supp}_\mathrm{R}}

\newcommand{\maxIdeal}{\mathfrak{m}}

\title{Low-Rank Parity-Check Codes over the Ring of Integers Modulo a Prime Power}

\author{Julian Renner}
\author{\IEEEauthorblockN{Julian Renner$^1$, Sven Puchinger$^2$, Antonia Wachter-Zeh$^1$, Camilla Hollanti$^3$, Ragnar Freij-Hollanti$^3$}
\IEEEauthorblockA{
$^1$Institute for Communications Engineering, Technical University of Munich (TUM), Germany\\
$^2$Department of Applied Mathematics and Computer Science, Technical University of Denmark (DTU), Denmark\\
$^3$Department of Mathematics and Systems Analysis, Aalto University, Finland\\
Email: julian.renner@tum.de, svepu@dtu.dk, antonia.wachter-zeh@tum.de, camilla.hollanti@aalto.fi, ragnar.freij@aalto.fi
\thanks{S. Puchinger has received funding from the European Union’s Horizon 2020 research and innovation programme under the Marie Skłodowska-Curie grant agreement no.~713683 (COFUNDfellowsDTU).\newline
\indent J. Renner and A. Wachter-Zeh were supported by the European Research Council (ERC) under the European Union’s Horizon 2020 research and innovation programme (grant agreement No 801434). \newline
\indent C. Hollanti was supported by the Academy of Finland, under grants 303819 and 318937, and by the TU Munich -- Institute for Advanced Study, funded by the German Excellence Initiative and the EU 7th Framework Programme under grant agreement no. 291763, via a Hans Fischer Fellowship.}
}}

\IEEEoverridecommandlockouts

\begin{document}

\maketitle

\begin{abstract}
We define and analyze low-rank parity-check (LRPC) codes over extension rings of the finite chain ring $\ZZ_{p^r}$, where $p$ is a prime and $r$ is a positive integer.
 LRPC codes have originally been proposed by Gaborit \emph{et al.}~(2013) over finite fields for cryptographic applications.
The adaption to finite rings is inspired by a recent paper by Kamche \emph{et al.}~(2019), which constructed Gabidulin codes over finite principle ideal rings with applications to space-time codes and network coding.
We give a decoding algorithm based on simple linear-algebraic operations.
Further, we derive an upper bound on the failure probability of the decoder. 
The upper bound is valid  for errors whose rank is equal to the free rank.
\end{abstract}

\section{Introduction}

Low-rank parity check (LRPC) codes were introduced over finite fields in \cite{gaborit2013low} and are rank-metric codes with applications to cryptography \cite{gaborit2013low}, powerline communications \cite{yazbek2017LRPCPowerLine}, and network coding \cite{8377229}.
They can be seen as the rank-metric analogs of low-density parity-check codes in the Hamming metric.
In \cite{aragon2019low}, new decoders for LRPC codes were proposed.
Compared to other known rank-metric codes, LRPC codes have a comparably low minimum distance, but their decoding is efficient and they have a weak algebraic structure. The latter property makes them suitable for cryptography: cryptosystems based on LRPC codes \cite{melchor2019rollo} are among the most promising candidates for future public-key encryption and key encapsulation systems that are secure against attacks by quantum computers. They achieve small public key sizes compared to other code-based systems and are supported by strong security reductions.

The rank metric and most of the known rank-metric codes \cite{de78,ga85a, ro91, sheekey2019mrd}  have been originally defined over finite fields.
Recently, \cite{kamche2019rank} studied rank-metric codes over finite principal ideal rings and defined, analyzed and proposed a decoder for Gabidulin codes over these rings. They also studied applications to network coding and space-time coding, where the codes over finite
rings have advantages compared to rank-metric codes over finite
fields.

In this paper, we combine the ideas of \cite{gaborit2013low} and \cite{kamche2019rank} by studying LRPC codes over the finite chain ring $\ZZ_{p^r}$, where a finite chain ring is a ring whose ideals are linearly ordered by inclusion. We describe a decoder that is similar to \cite{gaborit2013low} and analyze its failure probability for error vectors whose rank is equal to the free rank. This limitation is acceptable since in applications like the McEliece cryptosystem the errors can be restricted to such vectors. 
Similar to \cite{gaborit2013low}, the main difficulty is the derivation of a bound on the failure probability, which becomes even more involved when replacing fields by rings.

The results constitute a proof of concept that LRPC codes work also over finite rings.
Similar to Gabidulin codes over rings, the new codes can be applied to network coding and space-time coding,  
cf.~\cite{kamche2019rank}. The benefit of ring LRPC codes compared to ring Gabidulin codes is a potentially faster and simpler decoder, which comes at the cost of a small failure probability.
Furthermore, these codes can 
be considered for code-based cryptography, where replacing a field by a finite ring might increase the cost of generic decoding attacks.
Studying these applications in detail is out of the scope of this paper and should be done in future work.
It would also be interesting, especially in the context of cryptography, to extend the codes and the decoder failure bound to a wider class of finite rings.

\section{Preliminaries}

We use a similar notation and the properties of rings stated in \cite{kamche2019rank}. 
Let $p$ be a prime, $r$ and $m$ positive integers, $q \! = \!p^r$, $\Rq \! = \! \ZZ_{q}$, and $\Rqm \! =  \! \Rq[x]/(h)$, where $h \! \in \! \Rq[x]$ is a monic polynomial of degree $m$, which, if projected to $\mathbb{F}_p[x]$\footnote{\vspace{-0.01cm}Note that one can map every element of $\mathbb{F}_p$ to an arbitrary element of the residue class $\Rq/\maxIdeal$, where $\maxIdeal$ is maximal ideal of $\Rq$.}, is irreducible over $\mathbb{F}_p$. Note that elements in $\Rqm$ can be seen as vectors in $\Rq^m$.

We denote the set of $m\times n$ matrices over a ring $R$ by $R^{m\times n}$ and the set of row vectors of length $n$ over $R$ by $R^{n} = R^{1\times n}$. Rows and columns of $m\times n$ matrices are indexed by $1,\hdots,m$ and $1,\hdots,n$, where $A_{i,j}$ is the entry in the $i$-th row and $j$-th column of the matrix $\A$. For all $\A \in R^{m \times n}$, there exist an invertible matrix $\S \in R^{m\times m}$, an invertible matrix $\T \in R^{n\times n}$ and a diagonal matrix $\D \in R^{m\times n}$ such that $\D = \S \A \T$, where $\D$ is called a Smith normal form of $\A$. The rank and the free rank of $\A$ is defined by $\rk (\A) := |\{ i\in\{1,\hdots,\min\{m,n\}\}: \D_{i,i} \not = 0 \}|$ and $\frk (\A) := |\{ i \in \{1,\hdots,\min\{m,n\}\} :\D_{i,i} \text{ is a unit} \}|$, respectively.

Let $\ve{\gamma}=[\gamma_1,\hdots,\gamma_m]$ be an ordered basis of $\Rqm$ over $\Rq$. By utilizing the vector space isomorphism $\Rqm \cong \Rq^m$ , we can relate each vector $\a \in \Rqm^{n}$ to a matrix $\A \in \Rq^{m\times n}$ according to $\extsmallfield_{\gamma} : \Rqm^{n} \rightarrow \Rq^{m\times n}, \a \mapsto \A$, where $a_j = \sum_{i=1}^{m} A_{i,j} \gamma_{i}$,\ $j \in \{1,\hdots,n\}$. 
Note that $\Rqm$ is a ring of $q^m$ elements and a free $\Rq$-module of dimension $m$. Hence, elements of $\Rqm$ can be treated as vectors in $\Rq^m$ and linear independence, $\Rq$-subspaces of $\Rqm$ and the $\Rq$-linear span of elements are well-defined. 

\begin{lemma}[{\!\!\cite[Lemma~2.4]{kamche2019rank}}]\label{lem:Rqm_units}
Let $x \in \Rqm$. Then, $x$ is linearly independent over $\Rq$ if and only if $x$ is a unit in $\Rqm$.
\end{lemma}

The lemma above implies that $x \in \Rqm$ is a unit if and only if at least one entry of its vector representation is a unit.
Note that we have $|\Rq^*| = q(1-1/p)$, so $|\Rqm^*| = q^m-(|\Rq|-|\Rq^*|)^m = q^m[1-(1/p)^m]$.

\begin{lemma}\label{lem:Rq_decomposition}
For any $x \in \Rq \setminus \{0\}$, there is a unique integer $j \in \{0,\dots,r-1\}$ such that $ \exists \, x^* \in \Rq^*$ with
$x = p^j x^*$.
\end{lemma}

\begin{proof}
Trivial since $\Rq$ is a finite chain ring, the integers $p^0,\dots,p^{r-1}$ generate the ideals of the ring, and $j$ is the largest integer such that $x\in(p^j)$.
\end{proof}

\begin{lemma}\label{lem:Rqm_decomposition}
For any $x \in \Rqm \setminus \{0\}$, there is a unique integer $j \in \{0,\dots,r-1\}$ such that $ \exists \, x^* \in \Rqm^*$ with
$x = p^{j} x^*$.
\end{lemma}

\begin{proof}
The proof follows directly from Lemma~\ref{lem:Rqm_units} and Lemma~\ref{lem:Rq_decomposition} by choosing $j$ to be the minimum of the $j$'s of the entries of the vector representation of $x$ (this is independent of the  basis).
\end{proof}

The \emph{ (free) rank norm} of a vector $\a \in \Rqm^{n}$ is denoted by $(\ftemp)\rk_{\Rq}(\a)$ and is the (free) rank of the matrix representation $\A$, \emph{i.e.}, $\rk_{\Rq}(\a) := \rk(\A)$ and $\frk_{\Rq}(\a) := \frk(\A)$, respectively.

The $\Rq$-linear module that is spanned by $v_1,\hdots,v_{\ell} \in \Rqm$ is denoted by $\langle v_1,\dots,v_n \rangle_{\Rq} := \big\{\sum_{i=1}^{\ell} a_i v_i : a_i \in \Rq \big\}$. The $\Rq$-linear module that is spanned by the entries of a vector $\a \in \Rqm^{n}$ is called the support of $\a$, \emph{i.e.}, $\supp(\a) := \langle a_1,\dots,a_n \rangle_{\Rq}$. If a support has a basis, we refer to it as a free support. Further, $A B$ denotes the product module of two submodules $A$ and $B$ of $\Rqm$.

\section{LRPC Codes}

\begin{definition}\label{def:LRPCcodes}
Let $k,n,\lambda$ be positive integers with $0<k<n$. Furthermore, let $\Fspace \subseteq \Rqm$ be a free $\Rq$-submodule of $\Rqm$ of dimension $\lambda$.
A low-rank parity-check (LRPC) code with parameters $\lambda,n,k$ is a code with a parity-check matrix
$\H \in \Rqm^{(n-k) \times n}$
such that $\rank (\H) = n-k$ and $\Fspace = \langle H_{1,1},\dots,H_{(n-k),n} \rangle_{\Rq}$.
\end{definition}

Note that an LRPC code is a free submodule of $\Rqm^n$ of rank $k$.
We define the following three additional properties of the parity-check matrix that we will use throughout the paper to prove the correctness of our decoder and to derive failure probabilities.

\begin{definition}\label{def:H_properties}
Let $\lambda$, $\Fspace$, and $\H$ be defined as in Definition~\ref{def:LRPCcodes}.
Let $f_1,\dots,f_\lambda \in \Rqm$ be a free basis of $\Fspace$.
For $i=1,\dots,n-k$, $j=1,\dots,n$, and $\ell=1,\dots,\lambda$, let $h_{i,j,\ell} \in \Rq$ be the unique elements such that $H_{i,j} = \sum_{\ell = 1}^{\lambda} h_{i,j,\ell} f_{\ell}$.
Define
\begin{equation}
\H_{\mathrm{ext}} :=
\begin{bmatrix}
h_{1,1,1} & h_{1,2,1} & \hdots & h_{1,n,1} \\
h_{1,1,2} & h_{1,2,2} & \hdots & h_{1,n,2} \\
\vdots & \vdots & \ddots & \vdots \\
h_{2,1,1} & h_{2,2,1} & \hdots & h_{2,n,1} \\
h_{2,1,2} & h_{2,2,2} & \hdots & h_{2,n,2} \\
\vdots & \vdots & \ddots & \vdots \\
\end{bmatrix}
\in \Rq^{(n-k)\lambda \times n}.
\label{eq:H_ext}
\end{equation}
Then, $\H$ has the
\begin{enumerate}
\item \textbf{unique-decoding property} if 
$\lambda \geq \tfrac{n}{n-k}$ and $\frk \left( \H_{\mathrm{ext}} \right) = \rk \left(  \H_{\mathrm{ext}} \right) = n$,

\item \textbf{maximal-row-span property} if every row of the parity-check matrix $\H$ spans the entire space $\Fspace$,
\item \textbf{unity property} if every entry $H_{i,j}$ of $\H$ is chosen from the set
$H_{i,j} \in \tilde{\Fspace} := \left \{ \textstyle\sum_{i=1}^{\lambda} \alpha_i f_i \, : \, \alpha_i \in \Rq^* \cup \{0\} \right\}  \subseteq \Fspace$.
\end{enumerate}
\end{definition}

As its name suggests, the first property is related to unique erasure decoding, \emph{i.e.}, the process of obtaining the full error vector $\e$ after having recovered its support. The next lemma establishes this connection.

\begin{lemma}[Unique Erasure Decoding]\label{lem:erasure_decoding}
Given a parity-check matrix $\H$ that fulfills the {unique-decoding property}. Let $\Espace$ be a free support of dimension $t \leq \tfrac{m}{\lambda}$. If $\dim_{\Rq}(\Espace \Fspace) = \lambda t$, then, for any syndrome $\s \in \Rqm^{n-k}$, there is at most one error vector $\e \in \Rqm^n$ with support $\Espace$ that fulfills
$\H \e^\top = \s^\top$.
\end{lemma}

\begin{proof}
The proof follows by the same arguments as in \cite[Section~4.5]{aragon2019low} (see also \cite{renner2019efficient} for more details), where the {unique-decoding property} implies that $\H_\mathrm{ext}$ is full-rank and shows that there is at most one solution $\B \in \Rq^{t \times n}$ that solves the linear system of equations $\H \B^\top \a^\top = \s$, where the entries of $\a \in \Rqm^t$ form a basis of $\Espace$.
\end{proof}

In the original papers about LRPC codes over finite fields, \cite{gaborit2013low,aragon2019low}, some of the properties of Definition~\ref{def:H_properties} are used without explicitly stating them.

The {unique-decoding property} is necessary to obtain a unique decoding result after recovering the support of the error. Hence, the property is also necessary for the decoder in \cite{gaborit2013low} to return a unique decoding result. In practice, this property is not very restrictive: for  entries $H_{i,j}$ chosen uniformly at random from $\Fspace$, this property is fulfilled with the probability that a random $\lambda (n-k) \times n$ matrix has {full (free) rank $n$, which again is arbitrarily close to $1$ for increasing $\lambda(n-k)-n$ (cf.~\cite{renner2019efficient} for the field and Lemma~\ref{lem:syndrome_number_of_full-rank_matrices} for the ring case).} 

We will use the {maximal-row-span property} to prove a bound on the failure probability of the decoder in Section~\ref{sec:failure}. It is a sufficient condition that our bound (in particular Theorem~\ref{thm:syndrome_condition_main_statement} in Section~\ref{sec:failure}) holds. Although not explicitly stated, \cite[Proposition~4.3]{aragon2019low} must also assume a similar or slightly weaker condition in order to hold. It does not hold for arbitrary parity-check matrices as in \cite[Definition~4.1]{aragon2019low} (see Remark~\ref{rem:necessity_of_maximal_row_span_condition_or_similar} in Section~\ref{sec:failure}). This is again not a big limitation in general for two reasons: first, the ideal codes in \cite[Definition~4.2]{aragon2019low} appear to automatically have this property, and second, a random parity-check matrix has this property with high probability.

In the case of finite fields, the {unity property} is no restriction at all since the units of a finite field are all non-zero elements. That is, we have $\tilde{\Fspace} = \Fspace$. Over rings, we need this additional property as a sufficient condition for one of our failure probability bounds (Theorem~\ref{thm:syndrome_condition_main_statement} in Section~\ref{sec:failure}). It is not a severe restriction in general, since $|\tilde{\Fspace}| = (|\Rq^*|+1)^\lambda = (\tfrac{q}{p}+1)^\lambda = (p^{r-1}+1)^\lambda$ compared to $|\Fspace| = q^\lambda = p^{r\lambda}$.

\section{Decoding}

Fix $\lambda$ and $\Fspace$ as in Definition~\ref{def:LRPCcodes}.
Let $f_1,\dots,f_\lambda \in \Rqm$ be a free basis of $\Fspace$.
Note that since the $f_i$ are linearly independent, the sets $\{f_i\}$ are linearly independent, which by the above discussion implies that all the $f_i$ are units in $\Rqm$. Hence, $f_i^{-1}$ exists for each $i$.

\begin{algorithm}
\DontPrintSemicolon
\caption{LRPC Decoder}
\label{alg:decoder}
\KwIn{\begin{itemize}
\item LRPC parity-check matrix $\H$ (as in Definition~\ref{def:LRPCcodes})
\item $\r = \c + \e$, such that
\begin{itemize}
\item $\c$ is in the LRPC code $\mathcal{C}$ given by $\H$ and
\item The support of $\e$ is a free module of dimension $t$. 
\end{itemize}
\end{itemize}}
\KwOut{Codeword $\c'$ of $\mathcal{C}$ or ``decoding failure''}
$\s = [s_1,\dots,s_{n-k}] \gets \r \H^\top$ \\
$\Sspace \gets \langle s_1,\dots,s_{n-k} \rangle_{\Rq}$ \label{line:Sspace} \\
\If{$\dim_{\Rq} \Sspace < \lambda t$}{
\Return{``decoding failure''}
}
\For{$i=1,\dots,\lambda$}{
$\Sspace_i \gets f_i^{-1} \Sspace = \left\{f_i^{-1} s \, : \, s \in \Sspace \right\}$ \\
}
$\Espace' \gets \bigcap_{i=1}^{\lambda} \Sspace_i$ \\
\If{$\dim_{\Rq} \Espace' > t$ or $\dim_{\Rq} (\Espace' \Fspace) < \lambda t$}{
\Return{``decoding failure''}
}
$\e \gets$ Erasure decoding with support $\Espace'$ w.r.t.~the syndrome $\s$, as described in Lemma~\ref{lem:erasure_decoding} (analogous to \cite[Section~4.5]{aragon2019low} or \cite[Section~III.B]{renner2019efficient}) \\
\Return{$\r-\e$}
\end{algorithm}

The following theorem states precisely under which conditions on the error support and parity-check matrix space $\Fspace$ the decoder (Algorithm~\ref{alg:decoder}) returns the transmitted codeword.
For fixed $\Fspace$ and random errors of a given weight $t$, we study the probability of failure (\emph{i.e.}, the probability that the conditions are not fulfilled) in Section~\ref{sec:failure}.

\begin{theorem}\label{thm:decoder_correctness}
Let $\H$ be chosen as in Definition~\ref{def:LRPCcodes} such that it has the {unique-decoding property} (cf.~Definition~\ref{def:H_properties}).
Then, Algorithm~\ref{alg:decoder} returns the correct codeword $\c$ if the following three conditions are fulfilled:
\begin{enumerate}
\item \label{itm:fail_Sspace_small} 	$\dim_{\Rq} \Sspace = \lambda t$, \hfill (\textbf{syndrome condition}),
\item \label{itm:fail_intersection_big} $\dim_{\Rq} \left( \bigcap_{i=1}^{\lambda} \Sspace_i \right) = t$, \hfill (\textbf{intersection condition}),
\item \label{itm:fail_product_small} 	$\dim_{\Rq} \left( \Espace \Fspace \right) = \lambda t$, \hfill (\textbf{product condition}).	
\end{enumerate}
\end{theorem}

\begin{proof}
In Line~\ref{line:Sspace}, Algorithm~\ref{alg:decoder} computes the module spanned by the syndrome. Since the syndromes are sums of products of error and parity-check matrix entries, the syndrome space $\Sspace$ is a subset of the product space $\Espace \Fspace$. Due to the {syndrome condition}, we have equality, \emph{i.e.},
$\Sspace = \Espace \Fspace$.

By the definition of $\Sspace_1,\hdots,\Sspace_\lambda$, we have that $\Espace \subseteq \Sspace_i$ for $i=1\hdots,\lambda$ and thus $\Espace \subseteq \bigcap_{i=1}^{\lambda} \Sspace_i$.
Due to the {intersection condition}, the space $\bigcap_{i=1}^{\lambda} \Sspace_i$ cannot be larger than $\Espace$ and we have equality, i.e,
$\Espace = \bigcap_{i=1}^{\lambda} \Sspace_i$.
The {product condition} on the error, together with the {unique-decoding property} of the parity-check matrix ensures that we can recover uniquely the error vector $\e$ from its support (cf.~Lemma~\ref{lem:erasure_decoding}).
\end{proof}

\begin{remark}
Note that the conditions in Theorem~\ref{thm:decoder_correctness} imply that $\lambda t \leq m$ (due to the {product condition}) as well as $\lambda \geq \tfrac{n}{n-k}$ (due to the {unique-decoding property}). Combined, we obtain
$t \leq m\tfrac{n-k}{n} = m(1-R)$,
where $R := \tfrac{k}{n}$ is the rate of the LRPC code.
\end{remark}

\section{Failure Probability}\label{sec:failure}

\subsection{Failure of Product Condition}

\begin{lemma}\label{lem:product_space_full_dimension_probability_induction}
Let $A',B$ be free submodules of $\Rqm$ of dimension $\alpha'$ and $\beta$, respectively, such that also $A'B$ is a free submodule of $\Rqm$ of dimension $\alpha'\beta$. For an element $a \in \Rqm^\ast$, chosen uniformly at random, let $A := A' + \langle a \rangle$. Then, we have
\begin{align*}
&\Pr\big(\text{$AB$ is a free module of dimension $\alpha'\beta+\beta$}\big) \\
&\geq 1- \textstyle \sum_{j = 0}^{r-1} \Big[(q/p^j)^\beta-(q/p^{j+1})^\beta\Big] \left(q/p^j\right)^{\alpha' \beta-m}.
\end{align*}
\end{lemma}

\begin{proof}
First note that since $a$ is a unit in $\Rqm$, the mapping
$\varphi_a \, : \, B \to \Rqm, ~
b \mapsto ab$
is injective. This means that $aB$ is a free module of dimension $\dim_{\Rq}(aB)=\dim_{\Rq}(B)=\beta$. Let $b_1,\dots,b_\beta$ be a basis of $B$.
Then, $a b_1, \dots, a b_\beta$ is a basis of $aB$.

Hence, $AB$ is a free module of dimension $\dim(AB) = \alpha \beta+\beta$ if and only if all the elements $a b_1, \dots, a b_\beta$ are linearly independent of $A'B$. This again holds if and only if
$\sum_{i=1}^{\beta} \lambda_i a b_i \notin A'B \quad \forall \lambda_i \in \Rq, \text{ not all $0$}$.
This is equivalent to
$aB \cap A'B = \{0\}$.
Hence,
\begin{equation}
\Pr \! \big( \! \dim \! AB \! \not = \! \alpha'\beta \!+ \!\beta \big) \leq \Pr\left( \exists b \in \! B  \setminus \{0\} \! : ab \in A'B \right). \label{eq:product_inequality_1}
\end{equation}
Let $c$ be chosen uniformly at random from $\Rqm$. Recall that $a$ is chosen uniformly at random from $\Rqm^*$. Then,
\begin{equation}
\Pr \! \left(  \exists b \! \in \! B \! \setminus \! \{0\} \! : \! ab \! \in \! A' \!B  \! \right)
\! \leq \! \Pr \! \left( \exists b \! \in \! B \! \setminus \! \{0\} \! : \! cb \! \in \! A'\!B  \right). \label{eq:product_inequality_2}
\end{equation}
This holds since if $c$ is chosen to be a non-unit in $\Rqm$, then the statement ``$\exists \, b \in B \setminus \{0\} \, : \, cb \in A'B$'' is always true.
To see this, write $c = p c'$ for some $c' \in \Rqm$. Since $\beta>0$, there is a unit $b^* \in B \cap \Rqm^*$.
Choose $b := p^{r-1}b^* \in B\setminus \{0\}$. Hence, $c b = p c' p^{r-1}b^* = 0$, and $b$ is from $B$ and non-zero.

Now we bound the right-hand side of \eqref{eq:product_inequality_2} as follows
\begin{align*}
\Pr\left( \exists  b \in \! B \! \setminus \! \{0\} \! : \! cb \in \! A'B \right) \! &\leq \textstyle \sum_{b \in B \setminus \{0\}} \! \Pr \left( cb \in A'B \right) \\
&= \! \sum_{j = 0}^{r-1}  \sum_{b \in B :  j_b = j} \! \Pr\left( cb^* p^{j} \! \in \! A'B \right).
\end{align*}
Since $b^*$ is a unit in $\Rqm$, for uniformly drawn $c$, $c b^*$ is also uniformly distributed on $\Rqm$. Hence, $cb^* p^{j}$ is uniformly distributed on the ideal $p^{j}\Rqm$ of $\Rqm$ generated by $p^{j}$ and we have
$\Pr\left( cb^* p^{j} \in A'B \right) = \frac{\left| p^{j}\Rqm \cap A'B \right|}{|p^{j}\Rqm|}$.
Let $v_1,\dots,v_{\alpha'\beta}$ be a basis of $A'B$. Then, an element $c \in A'B$ is in $p^{j}\Rqm$ if and only if it can be written as
$c = \sum_{i} \mu_i v_i$,
where $\mu_i \in p^j \Rq$ for all $i$. This is true due to the following argument: Assume not. Then there is a non-empty set $\mathcal{I} \subseteq \{1,\dots,\alpha\beta\}$ such that $\mu_i \notin p^j \Rq$ for all $i \in \mathcal{I}$ and $\mu_i \in p^j \Rq$ for all $i \notin \mathcal{I}$. Note that this implies $p^{r-j}\mu_i \neq 0$ if and only if $i \in \mathcal{I}$. Hence,
$0 = p^{r-j}c = \sum_{i \in \mathcal{I}} p^{r-j}\mu_i v_i$.
However, this contradicts the fact that the $v_i$ are linearly independent since all the $p^{r-j}\mu_i$ are in $\Rq$, but not zero.

Hence, 
$\left| p^{j}\Rqm \cap A'B \right| = |p^{j}\Rq|^{\alpha' \beta}$.
Furthermore, we have
$|p^{j}\Rqm| = |p^{j}\Rq|^m$,
where
$|p^{j}\Rq| = q/p^j$.
Overall, we get
\begin{align}
&\Pr\left( \exists \, b \in  B  \setminus \{0\} \, : \, cb  \in  A'B \right) \notag \\
&\leq \textstyle \sum_{j = 0}^{r-1} \textstyle \sum_{b \in B \, : \, j_b = j} \left(q/p^j\right)^{\alpha' \beta-m} \notag \\ 
&= \textstyle \sum_{j = 0}^{r-1} \big|\{b \in B \, : \, j_b = j\}\big| \left(q/p^j\right)^{\alpha' \beta-m}. \label{eq:product_inequality_3}
\end{align}
Furthermore, we have (note that $p^{j+1}\Rqm \subseteq p^{j}\Rqm$)
\begin{align}
\big|\{b \in B \, : \, j_b = j\}\big| &= \Big|\big(p^{j}\Rqm \setminus p^{j+1}\Rqm\big) \cap B \Big| \notag\\
&= \big|p^{j}\Rqm \cap B \big| - \big|p^{j+1}\Rqm \cap B \big| \notag \\
&= (q/p^j)^\beta-(q/p^{j+1})^\beta. \label{eq:product_inequality_4}
\end{align}
Combining \eqref{eq:product_inequality_1}, \eqref{eq:product_inequality_2}, \eqref{eq:product_inequality_3}, and \eqref{eq:product_inequality_4} gives the result.
\end{proof}
\begin{lemma}\label{lem:product_space_full_dimension_probability}
Let $B$ be a fixed free submodule of $\Rqm$ of dimension $\beta$. Furthermore, let $A$ be drawn uniformly at random from the set of free submodules of $\Rqm$ of dimension $\alpha$.
Then,
\vspace{-0.005cm}
\begin{align*}
&\Pr\big(\text{$AB$ is a free module of dimension $\alpha\beta$}\big) \\
&\geq 1- \alpha \textstyle \sum_{j = 0}^{r-1} \Big[(q/p^j)^\beta-(q/p^{j+1})^\beta\Big] \left(q/p^j\right)^{\alpha \beta-m}.
\end{align*}
\end{lemma}
\vspace{-0.15cm}
\begin{proof}
Drawing a free submodule $A \subseteq \Rqm$ of dimension $\alpha$ uniformly at random is equivalent to drawing iteratively
$A_0 := \{0\}, ~
A_i := A_{i-1} + \langle a_i \rangle$ for $i=1,\dots,\alpha$
where for each iteration $i$, the element $a_i \in \Rqm$ is chosen uniformly at random from the set of vectors that are linearly independent of $A_{i-1}$. The equivalence of the two random experiments is clear since the possible choices of the sequence $a_1,\dots,a_\alpha$ gives exactly all bases of free $\Rq$-submodules of $\Rqm$ of dimension $\alpha$. Furthermore, all sequences are equally likely and each resulting submodule has the same number of bases that generate it (which equals the number of invertible $\alpha \times \alpha$ matrices over $\Rq$).
We have the following recursive formula for any $i=1,\dots,\alpha$:
\begin{align*}
&\Pr\big( \dim(A_i B)< i \beta \big) \\
&= \Pr\big( \dim(A_i B)< i \beta \land \dim(A_{i-1}B)=(i-1)\beta\big) \\
&\quad + \underbrace{\Pr\big( \dim(A_i B)< i \beta \land \dim(A_{i-1}B)<(i-1)\beta\big)}_{\text{$\dim(A_{i-1}B)<(i-1)\beta$ implies $\dim(A_i B)< i \beta$}} \\
&= \underbrace{\Pr\big( \dim(A_i B)< i \beta  \mid \dim(A_{i-1}B)=(i-1)\beta\big)}_{\overset{(\ast)}{\leq} \, \sum_{j = 0}^{r-1} \Big[(q/p^j)^\beta-(q/p^{j+1})^\beta\Big] \left(q/p^j\right)^{(i-1) \beta-m}} \\ &\quad \quad \cdot \underbrace{\Pr(\dim(A_{i-1}B)=(i-1)\beta)}_{\leq 1} \\
& \quad + \Pr\big(\dim(A_{i-1}B)<(i-1)\beta\big) \\
&\leq \textstyle \sum_{j = 0}^{r-1}\Big[(q/p^j)^\beta-(q/p^{j+1})^\beta\Big] \left(q/p^j\right)^{(i-1) \beta-m} \\
&\quad + \Pr\big(\dim(A_{i-1}B)<(i-1)\beta\big),
\end{align*}
where ($\ast$) follows from Lemma~\ref{lem:product_space_full_dimension_probability_induction} by the following additional argument:
\begin{align*}
&\Pr\big( \dim(A_i B)< i \beta  \mid \\
&\quad  \quad \dim(A_{i-1}B)=(i-1)\beta \land a_i \text{ l.i.\ of $A_{i-1}$}\big) \\
&\leq \Pr\big( \dim(A_i B)< i \beta  \mid \\
& \quad \quad \dim(A_{i-1}B)=(i-1)\beta \land a_i \text{ uniformly from } \Rqm^* \big) \\
&\leq \sum_{j = 0}^{r-1} \Big[(q/p^j)^\beta-(q/p^{j+1})^\beta\Big] \left(q/p^j\right)^{(i-1) \beta-m},
\end{align*}
where the last inequality is exactly the statement of Lemma~\ref{lem:product_space_full_dimension_probability_induction}.

By $\Pr\big(\dim(A_{0}B)<0\big) = 0$, we get
\begin{align*}
&\Pr\left( \dim(AB)< \alpha \beta \right) \\
&= \Pr\big( \dim(A_\alpha B)< \alpha \beta \big) \\
&= \textstyle \sum_{i=1}^{\alpha} \textstyle \sum_{j = 0}^{r-1} \Big[(q/p^j)^\beta-(q/p^{j+1})^\beta\Big] \left(q/p^j\right)^{(i-1) \beta-m} \\
&\leq \alpha \textstyle \sum_{j = 0}^{r-1} \Big[(q/p^j)^\beta-(q/p^{j+1})^\beta\Big] \left(q/p^j\right)^{\alpha \beta-m}.
\end{align*}
This proves the claim.
\end{proof}

The following theorem follows directly from the previous lemma by choosing $A$ to be the random error support of dimension $t$ and $B$ to be the fixed submodule $\Fspace$ of dimension~$\lambda$.

\begin{theorem}\label{thm:product_condition_main_statement}
Let $\Fspace$ be defined as in Definition~\ref{def:LRPCcodes}. Let $t$ be a positive integer with $t \lambda <m$ and let $\Espace$ be the support of an  error word $\e$ chosen uniformly at random among all error words with free support of dimension $t$.
Then, the probability that the {product condition} is not fulfilled is
\begin{align*}
&\Pr\left( \dim_{\Rq} \left( \Espace \Fspace \right) < \lambda t \right) \\
&\leq t \textstyle \sum_{j = 0}^{r-1} \Big[(q/p^j)^\lambda-(q/p^{j+1})^\lambda\Big] \left(q/p^j\right)^{\lambda t-m}.
\end{align*}
\end{theorem}

\subsection{Failure of Syndrome Condition}

\begin{lemma}\label{lem:syndrome_number_of_full-rank_matrices}
Let $a,b$ be positive integers with $a < b$. Then,
$\NM(a,b;\Rq) := |\{\A \in \Rq^{a\times b}:\frk(\A)=\rk(\A) =a \}| = q^{a b} \prod_{a'=0}^{a-1} \left(1-p^{a'-b} \right)$.
\end{lemma}

\begin{proof}
First note that $\NM(1,b;\Rq) = q^b-(q/p)^b$ since a $1 \times b$ matrices over $\Rq$ is of free rank $1$ if and only if at least one entry is a unit. Hence we subtract from the number of all matrices ($q^b$) the number of vectors that consist only of non-units ($(q/p)^b$).

Let now for any $a' \leq a$ be $\A \in \Rq^{a' \times b}$ a matrix of free rank~$a'$. We define 
$\mathcal{V}(\A) := \big\{ \v \in \Rq^{1 \times b} \! : \! \frk \big(\begin{bmatrix}
\A^\top 
\v^\top
\end{bmatrix}^\top \big) = a' \big\}$.
We study the cardinality of $\mathcal{V}(\A)$.
We have $\frk\big(\begin{bmatrix}
\A^\top
\v^\top
\end{bmatrix}^\top\big) = a'$ if and only if the rows of the matrix $\hat{\A} := \begin{bmatrix}
\A^\top
\v^\top
\end{bmatrix}^\top$ are linearly dependent.
Due to $\frk(\A) = a'$ and the existence of a Smith normal form of $\A$, there are invertibe matrices $\S$ and $\T$ such that
$\S \A \T = \D$,
where $\D$ is a diagonal matrix with ones on its diagonal.

Since $\S$ and $\T$ are invertible, we can count the number of vectors $\v'$ such that the rows of the matrix
$\big[
\D^\top
{\v'}^\top
\big]^\top$
are linearly independent instead of the matrix $\hat{\A}$ (note that $\v = \v' \T^{-1}$ gives a corresponding linearly dependent row in~$\hat{\A}$).

Since $\D$ is in diagonal form with only ones on its diagonal, the linearly dependent vectors are exactly of the form
\begin{equation*}
\v' = [v'_1,\dots,v'_a,v'_{a'+1},\dots, v'_b],
\end{equation*}
where $v'_i \in \Rq$ for $i=1,\dots,a'$ and $v'_i \in p \Rq$ for $i=a'+1,\dots,b$. Hence, we have 
$|\mathcal{V}(\A)| = q^{a'} (q/p)^{b-a'} = q^b p^{a'-b}$.
Note that this value is independent of $\A$.

By the discussion on $|\mathcal{V}(\A)|$, we get the following recursive formula:
\begin{align*}
\NM(a'\!+ \!1,b;\Rq) \! = \! \begin{cases}
\NM(a',b;\Rq)\! \left( q^b-q^b p^{a'-b} \right), \! \! &a'\geq 1, \\
q^b-(q/p)^b, &a'=0,
\end{cases}
\end{align*}
which resolves into
$\NM(a,b;\Rq) \! \! = \! \!  q^{a b} \prod_{a'=0}^{a-1} \left(1-p^{a'-b}\right)$.
\end{proof}

\begin{theorem}\label{thm:syndrome_condition_main_statement}
Suppose that the product condition is fulfilled.
Let $\Fspace$ be defined as in Defintion~\ref{def:LRPCcodes}. Let $t$ be a positive integer with $t \lambda <m$ and $\Espace$ be the support of a error word $\e$ chosen uniformly at random among all error words with free support of dimension $t$.

Suppose further that $\H$ has the {maximal-row-span} and {unity properties} (cf.~Definition~\ref{def:H_properties}).

Then, the probability that the {syndrome condition} is not fulfilled is
\begin{align*}
&\Pr\left( \dim_{\Rq} \left( \Sspace \right) < \lambda t  \mid \dim_{\Rq} \left( \Espace \Fspace \right) = \lambda t \right) \\
&\leq 1 - \textstyle \prod_{i=0}^{\lambda t -1} \left(1-p^{i-(n-k)}\right).
\end{align*}
\end{theorem}

\begin{proof}
Let $\e' \in \Rqm^n$ be chosen such that every entry $e_i'$ is chosen uniformly at random from the error support $\Espace$.\footnote{This means that $\e'$ might have rank $t$ or smaller. The difference to the actual error $\e$ is that $\e$ is chosen uniformly at random from the vectors of rank $t$.}
Denote by $\S_{\e}$ and $\S_{\e'}$ the syndrome spaces obtained by computing the syndromes of $\e$ and $\e'$, respectively.
Then, we have
\begin{align*}
\Pr\big(\Sspace_{\e'} = \Espace \Fspace\big) &\leq \Pr\big( \Sspace_{\e'} = \Espace \Fspace \mid \supp(\e') = \Espace \big) \\
&= \Pr\big( \Sspace_{\e} = \Espace \Fspace \big),
\end{align*}
where the latter equality follows from the fact that the random experiments of choosing $\e'$ and conditioning on the property that $\e'$ has free rank $t$ is the same as directly drawing $\e$ uniformly at random from the set of free rank $t$ errors.
Hence, we obtain a lower bound on $\Pr\big( \Sspace_{\e} = \Espace \Fspace \big)$ by studying $\Pr\big(\Sspace_{\e'} = \Espace \Fspace\big)$, which we do in the following.

Let $f_1,\dots,f_\lambda$ be a basis of $\Fspace$ and $\varepsilon_1,\dots,\varepsilon_t$ be a basis of $\Espace$.
Since $e'_i$ is an element drawn uniformly at random from $\Espace$, we can write it as
$e'_i = \sum_{\mu=1}^{t} e_{i,\mu}' \varepsilon_\mu$,
where $e_{i,j}'$ are uniformly distributed on $\Rq$.
Furthermore, we can write any $H_{i,j}$ as
$H_{i,j} = \sum_{\eta=1}^{\lambda} h_{i,j,\eta} f_\eta$,
where the $h_{i,j,\eta}$ are units in $\Rq$ or zero (due to the {unity property}). Furthermore, since each row of $\H$ spans the entire module $\Fspace$ ({full-row-span property}), for each $i$ and each $\eta$, there is at least one $j^*$ with $h_{i,j^*,\eta}$. By the previous assumption, this means that $h_{i,j^*,\eta} \in \Rq^*$.

Then, each syndome coefficient can be written as
\begin{equation*}
s_i = \textstyle \sum_{j=1}^{n} e'_j H_{i,j}
= \textstyle \sum_{\mu=1}^{t} \textstyle \sum_{\eta=1}^{\lambda} \underbrace{\left(\textstyle \sum_{j=1}^{n}  e_{j,\mu}'  h_{i,j,\eta}\right)}_{=: s_{\mu,\eta,i}} \varepsilon_\mu f_\eta.
\end{equation*}

By the above discussion, for each $i$ and $\eta$, there is a $j^*$ with $h_{i,j^*,\eta} \neq 0$. Hence, $s_{\mu,\eta,i}$ is a sum (with at least one summand) of the products of uniformly distributed elements of $\Rq$ and units of $\Rq$. A uniformly distributed ring element times a unit is also uniformly distributed on $\Rq$. Hence $s_{\mu,\eta,i}$ is a sum (with at least one summand) of uniformly distributed elements of $\Rq$. Hence, $s_{\mu,\eta,i}$ itself is uniformly distributed on $\Rq$.

All together, we can write
\begin{align*}
\begin{bmatrix}
s_1 \\
s_2 \\
\vdots \\
s_{n-k}
\end{bmatrix}
= 
\underbrace{
\begin{bmatrix}
s_{1,1,1} & s_{1,2,1} & \dots & s_{t,\lambda,1} \\
s_{1,1,2} & s_{1,2,2} & \dots & s_{t,\lambda,2} \\
\vdots & \vdots & \ddots & \vdots \\
s_{1,1,n-k} & s_{1,2,n-k} & \dots & s_{t,\lambda,n-k} \\
\end{bmatrix}}_{=: \, \S}
\cdot 
\begin{bmatrix}
\varepsilon_1 f_1 \\
\varepsilon_1 f_2 \\
\vdots \\
\varepsilon_t f_\lambda \\
\end{bmatrix},
\end{align*}
where the $\varepsilon_i f_j$ are a basis of $\Espace \Fspace$ (since the product condition is fulfilled by assumption) and the matrix $\S$ is chosen uniformly at random from $\Rq^{(n-k)\times t \lambda}$.
We have $\Sspace_{\e'} = \Espace\Fspace$ if and only if $\S$ has full free rank $t \lambda$.
By Lemma~\ref{lem:syndrome_number_of_full-rank_matrices}, the probability of drawing such a full-rank matrix is
$\frac{\NM(a,b;\Rq)}{q^{ab}} = \prod_{a'=0}^{a-1} \left(1-p^{a'-b} \right)$.
\end{proof}

\begin{remark}\label{rem:necessity_of_maximal_row_span_condition_or_similar}
In contrast to Theorem~\ref{thm:syndrome_condition_main_statement} the {full-row-span property} was not assumed in \cite[Proposition~4.3]{aragon2019low}, which is the analogous statement for finite fields.
However, also the statement in \cite[Proposition~4.3]{aragon2019low} is only correct if we assume additional structure on the parity-check matrix
{(\emph{e.g.}, that each row spans the entire space $\Fspace$ or a weaker condition), due to the following counterexample:
Consider a parity-check matrix $\H$ that contains only non-zero entries on its diagonal and in the last row.
More precisely, the diagonal entries are all $f_1$ and the last row contains the remaining $f_2,\dots,f_\lambda$.
This is a valid parity-check matrix according to \cite[Definition~4.1]{aragon2019low} since the entries of $\H$ span the entire space $\Fspace$.
However, due to the structure of the matrix, the first $n-k-1$ syndromes are all in $f_1 \Espace$, hence $\dim(\Sspace) \leq t+1 < t \lambda$ for \textbf{any} error of dimension $t$.}
\end{remark}

\subsection{Failure of Intersection Condition}

\begin{lemma}[{Equivalent of \cite[Lemma~3.4]{aragon2019low}}]\label{lem:intersection_failure_lemma}
Let $A,B \subseteq \Rqm$ be free $\Rq$-modules of dimensions $\alpha$ and $\beta$, respectively. Furthermore, let $\beta_2 := \dim(B^2)$.

Assume that $\dim(A B^2) = \alpha \beta_2$ and there is a module $E \subseteq \Rqm$ with $A \subsetneq E$ and $EB = AB$. Then, there is an $x \in B \setminus \Rq$ such that $xB \subseteq B$.
\end{lemma}

\begin{proof}
Let $a_1,\dots,a_\alpha$ be a basis of $A$ and $b_1,\dots,b_\beta$ be a basis of $B$.

First note that the existence of $E$ with the presumed properties implies that there is an $e \in AB \setminus A$ such that $eB \subseteq AB$.
Then, there are coefficients $e_{i,j} \in \Rq$ with
\begin{align}
\textstyle e = \sum_{i=1}^{\alpha} \underbrace{\left(\textstyle\sum_{j=1}^{\beta} e_{i,j} b_j \right)}_{=: \, b'_i} a_i. \label{eq:e_unique_representation}
\end{align}
By assumption, $e$ is not in $A$, which means that there is an $\eta \in \{1,\dots,\beta\}$ with $b_\eta' \notin \Rq$.

Let now $b \in B$. Since by assumption $eb \in AB$, there are $c_{i,j} \in \Rq$ with $e b = \sum_{i=1}^{\alpha} \left(\sum_{j=1}^{\beta} c_{i,j} b_j \right) a_i$.
By \eqref{eq:e_unique_representation}, we can also write 
$e b = \sum_{i=1}^{\alpha} \left(\sum_{j=1}^{\beta} e_{i,j} b_j b \right) a_i$.
Due to the maximality of the dimension of $AB^2$, there is a unique representation $c = \sum_i c_i a_i$ with $c_i \in B^2$ for each $c \in AB^2$. Since $eb \in AB$, we must therefore have 
$b_i' b = \left(\sum_{j=1}^{\beta} e_{i,j} b_j\right) b = \sum_{j=1}^{\beta} c_{i,j} b_j \quad \forall \, i$,
in particular
$b_\eta' b \in B$.
Since this hold for any $b$, we have $b_\eta' B \subseteq B$ (recall also that $b_\eta' \notin \Rq$). Choosing $x=b_\eta'$ gives the claimed result.
\end{proof}

\begin{theorem}\label{thm:intersection_failure_main_statement}
Suppose that the syndrome condition is fulfilled and $m$ is chosen such that the smallest intermediate ring $R'$ between $\Rq \subsetneq R' \subseteq \Rqm$ has cardinality greater than $q^\lambda$. 
Let $\Fspace$ be defined as in Defintion~\ref{def:LRPCcodes}. Let $t$ be a positive integer with $t \lambda <m$ and $\Espace$ be the support of a error word $\e$ chosen uniformly at random among all error words with free support of dimension $t$.

Then, the probability that the {intersection condition} is not fulfilled is
\begin{align*}
&\Pr\left( \dim_{\Rq} \left( \textstyle \bigcap_{i=1}^{\lambda} \Sspace_i \right) > t \mid   \dim_{\Rq} \left( \Sspace \right) = \lambda t  \right) \\
&\leq t \textstyle \sum_{j = 0}^{r-1} \Big[(q/p^j)^\beta-(q/p^{j+1})^\beta\Big] \left(q/p^j\right)^{\frac{t \lambda(\lambda+1)}{2}-m}.
\end{align*}
\end{theorem}

\begin{proof}
Assume that the intersection condition is not fulfilled. Then we have
$\bigcap_{i=1}^{\lambda} \Sspace_i =: \Espace' \supsetneq \Espace$.
Choose now $A = \Espace$, $E = \Espace'$, and $B = \Fspace$ in Lemma~\ref{lem:intersection_failure_lemma}. Since $\Espace$ is chosen uniformly at random from all free submodules of $\Rqm$ of dimension $t$, we can apply Lemma~\ref{lem:product_space_full_dimension_probability} and obtain that
$\dim(\Espace \Fspace^2) = t \lambda'$
with probability at least
\begin{align*}
&\Pr(\dim(\Espace \Fspace^2) = t \lambda') \\
&\geq 1- t \textstyle \sum_{j = 0}^{r-1} \Big[(q/p^j)^\beta-(q/p^{j+1})^\beta\Big] \left(q/p^j\right)^{t \lambda'-m} \\
&\geq 1- t \textstyle \sum_{j = 0}^{r-1} \Big[(q/p^j)^\beta-(q/p^{j+1})^\beta\Big] \left(q/p^j\right)^{\frac{t \lambda(\lambda+1)}{2}-m},
\end{align*}
where $\lambda' := \dim(\Fspace^2) \leq \tfrac{1}{2}\lambda(\lambda+1)$ (this is clear since $\Fspace^2$ is generated by the products of all unordered pairs of basis elements of $\Fspace$).

Hence, with probability at least this value, both conditions of Lemma~\ref{lem:intersection_failure_lemma} are fulfilled.
This means that there is an element $x \in \Fspace \setminus \Rq$ such that $x \Fspace \subseteq \Fspace$.
Thus, also $x^i \Fspace \subseteq \Fspace$ for all positive integers $i$, and we have that the ring $\Rq(x)$ extended by the element $x \notin \Rq$ fulfills $\Rq(x) \subseteq \Fspace$ (this holds since $\Fspace$ contains at least one unit).
By the condition on intermediate rings in the lemma statement, we must have $q^\lambda = |\Fspace| \geq |\Rq(x)| > q^\lambda$, a contradiction.
\end{proof}

\subsection{Overall Failure Probability}

\begin{theorem}\label{thm:main_statement}
Let $m$ be chosen such that the smallest intermediate ring $R'$ between $\Rq \subsetneq R' \subseteq \Rqm$ has cardinality greater than $q^\lambda$ and $\Fspace$ be defined as in Defintion~\ref{def:LRPCcodes}.
Suppose further that $\H$ has the {maximal-row-span} and {unity properties} (cf.~Definition~\ref{def:H_properties}).

Let $t$ be a positive integer with $t \lambda <m$ and $\e \in \Rqm^n$ be chosen uniformly at random from the set of vectors with free support of dimension $t$ (\emph{i.e.}, rank and free rank of $\e$ are $t$).

Then, Algorithm~\ref{alg:decoder} with input $\c+\e$ returns $\c$ with probability at least
\begin{align*}
&\Pr(\text{success}) \geq 1\\
&\quad - t \textstyle \sum_{j = 0}^{r-1} \Big[(q/p^j)^\beta-(q/p^{j+1})^\beta\Big] \left(q/p^j\right)^{\lambda t-m} \\
&\quad - \textstyle \prod_{i=0}^{\lambda t -1} \left(1-p^{i-(n-k)}\right) \\
&\quad - t \textstyle \sum_{j = 0}^{r-1} \Big[(q/p^j)^\beta-(q/p^{j+1})^\beta\Big] \left(q/p^j\right)^{\frac{t \lambda(\lambda+1)}{2}-m},
\end{align*}
independent of the transmitted codeword $\c$.
\end{theorem}

\begin{proof}
The statement follows by applying the union bound to the failure probabilities of the three success conditions, derived in Theorems~\ref{thm:product_condition_main_statement}, \ref{thm:syndrome_condition_main_statement}, and~\ref{thm:intersection_failure_main_statement}.
\end{proof}

\section{Simulation Results}
We performed simulations of LRPC codes with $\lambda=2$, $k=8$ and $n=20$ (note that we need $k \leq \tfrac{\lambda-1}{\lambda}n$ by the unique-decoding property) over the ring ${R}_{4,20}$ ($q=4$ and $m=20$). In each simulation, we generated one parity-check matrix (fulfilling the maximal-row-span and the unity properties) and conducted a Monte Carlo simulation in which we collected exactly $1000$ decoding errors. 
All simulations gave very similar results and confirmed our analysis. We show one of the simulation results in Figure~\ref{fig:sim} where we indicate by markers the estimated probabilities of violating the product condition (S: Prod), the syndrome condition (S: Synd), the intersection condition (S: Inter) as well as the decoding failure rate (S: Dec). Further we show the derived bounds on the probabilities of not fulfilling the product condition (B: Prod) given in Theorem~\ref{thm:product_condition_main_statement}, the syndrome condition (B: Synd) derived in Theorem~\ref{thm:syndrome_condition_main_statement}, the intersection condition (B: Inter) provided in Theorem~\ref{thm:intersection_failure_main_statement} and the union bound (B: Dec) stated in Theorem~\ref{thm:main_statement}. One can observe that the bound on the probability of not fulfilling the syndrome condition is very close to the true probability while the bounds on the probabilities of violating the product and syndrome condition are loose. Gaborit \emph{et al.} have made the same observation in the case of finite fields.

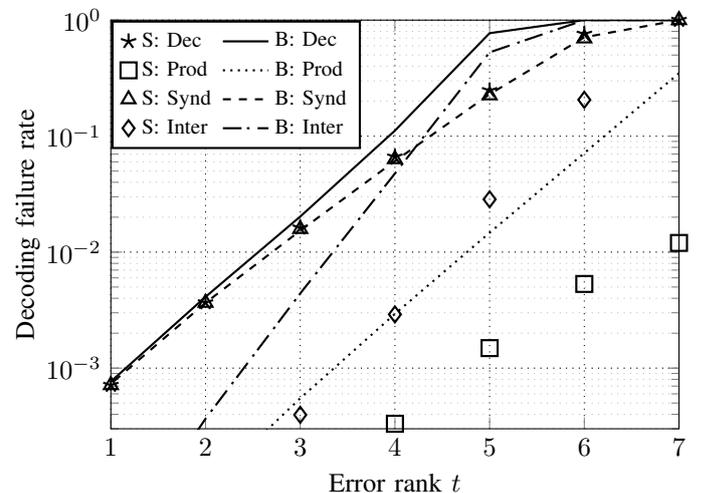
\begin{figure}[h]
  \begin{tikzpicture}
  \begin{semilogyaxis}[
legend columns=2, 
        legend style={
                    % the /tikz/ prefix is necessary here...
                    % otherwise, it might end-up with `/pgfplots/column 2`
                    % which is not what we want. compare pgfmanual.pdf
            /tikz/column 2/.style={
                column sep=2pt,
            },
          },
every axis plot/.append style={thick},
yticklabel style={/pgf/number format/fixed},
ylabel style={yshift=-1ex},
width=\columnwidth,
height=7.cm,
ymax = 1.0e-0,
ymin = 0.3e-3,
xmin = 1,
xmax = 7,
xtick={0,...,16},
%restrict x to domain=0:1,
ylabel={Decoding failure rate},
xlabel={Error rank $t$},
legend cell align=left,
grid=both,
minor grid style={dotted},
major grid style={dotted,black},
legend style={at={(0.001,0.999)},anchor=north west,thick,font={\footnotesize }},
mark options={fill=white,scale=1.5},
]
\newcommand\pLambda{2}
\newcommand\pK{16}
\newcommand\pN{2*\pK}
\newcommand\pM{30}
\newcommand\pQ{2}

% 1000 Error:
\newcommand{\file}{q4_m20_lamd2_n20_k8_id3104038726_results.txt}

%%%%%%%%%%%%%%%%%%%%%%%

%%%%%%%%%%%%%%%%%%%%%%%
\addplot[mark=star, black, only marks] table [x=t, y=fer, col sep=comma]{\file};
\addlegendentry{S: Dec}

\addplot[solid] table [x=t, y=bound, col sep=comma]{\file};
\addlegendentry{B: Dec}

\addplot[mark=square, black, only marks] table [x=t, y=fer1, col sep=comma]{\file};
\addlegendentry{S: Prod}

\addplot[dotted] table [x=t, y=bound1, col sep=comma]{\file};
\addlegendentry{B: Prod}

\addplot[mark=triangle, black, only marks] table [x=t, y=fer2, col sep=comma]{\file};
\addlegendentry{S: Synd}

\addplot[dashed] table [x=t, y=bound2, col sep=comma]{\file};
\addlegendentry{B: Synd}

\addplot[mark=diamond, black, only marks] table [x=t, y=fer3, col sep=comma]{\file};
\addlegendentry{S: Inter}

% \addplot[dash dot] table [x=t, y=bound3, col sep=comma]{results/\file};
\addplot[dash pattern={on 7pt off 2pt on 1pt off 3pt}] table [x=t, y=bound3, col sep=comma]{\file};

\addlegendentry{B: Inter}

\end{semilogyaxis}
\end{tikzpicture}	
%%% Local Variables:
%%% mode: latex
%%% TeX-master: "main"
%%% End:
  \vspace{-0.7cm}
\caption{Simulation results for $\lambda=2$, $k=8$ and $n=20$ over ${R}_{4^{20}}$. The markers indicate the estimated probabilities of not fulfilling the product condition (S: Prod), the syndrome condition (S: Synd), the intersection condition (S: Inter) and the decoding failure rate (S: Dec). The derived bounds on these probabilities are shown as lines.}
\label{fig:sim}
\end{figure}

\bibliographystyle{IEEEtran}
\bibliography{main}

\end{document}